\documentclass[aps,twocolumn,superscriptaddress,showpacs]{revtex4-1}
\usepackage{amsmath,amsfonts,amsbsy,amssymb,amsthm}

\usepackage{bbold,bbm,bm}
\usepackage{graphicx}
\usepackage{color}
\usepackage{hyperref}
\hypersetup{
    colorlinks=true,
    citecolor=black,
    linkcolor=black,
    filecolor=black,      
    urlcolor=black,
}

\theoremstyle{plain}
\newtheorem{theorem}{Theorem}
\newtheorem{corollary}[theorem]{Corollary}
\newtheorem{definition}[theorem]{Definition}
\newtheorem{proposition}[theorem]{Proposition}

\newcommand{\comment}[1]{}

\newcommand{\ket}[1]{| #1 \rangle}

\begin{document}

\title{Rigidity of the magic pentagram game}
\author{Amir Kalev}
\email{amirk@umd.edu}
\affiliation{Joint Center for Quantum Information and Computer Science, University of Maryland, College Park, MD 20742-2420, USA}
\author{Carl A. Miller}
\email{camiller@umd.edu}
\affiliation{Joint Center for Quantum Information and Computer Science, University of Maryland, College Park, MD 20742-2420, USA}
\affiliation{National Institute of Standards and Technology, Gaithersburg, MD 20899, USA}
\date{\today}

\begin{abstract}
A game is rigid if a near-optimal score guarantees, under the sole assumption of the validity of quantum mechanics, that the players are using an approximately unique quantum strategy.  Rigidity has a vital role in quantum cryptography as it permits a strictly classical user to trust behavior in the quantum realm.  This property can be traced back as far as 1998 (Mayers and Yao) and has been proved for multiple classes of games. In this paper we prove ridigity for the magic pentagram game, a simple binary constraint satisfaction game involving two players, five clauses and ten variables.  We show that all near-optimal strategies for the pentagram game are approximately equivalent to
a unique strategy involving real Pauli measurements on three maximally-entangled qubit  pairs.
\end{abstract}

\maketitle
\section{Introduction}\label{sec:intro}
Quantum rigidity is a strengthening of the guarantee that quantum behavior is taking place. It essentially ascertains that observing certain correlations in a system, for example, correlations that violate Bell inequalities, is sufficient by itself to determine the quantum state and the measurements used to obtain these correlations.   This notion was expressed in the work of Mayers and Yao on ``self-checking quantum sources''  \cite{mayers1998quantum} in 1998, and it can be traced back even earlier \cite{popescu1992states, Summers1987}.
Rigidity is a central tool for quantum computational protocols that involve untrusted devices,
since it allows a user to verify the internal workings of a device based only on its
external behavior (see, e.g.,  \cite{reichardt2013classical}).

Since its introduction the notion of rigidity has seen good deal of work, 
generally focused either on proving
rigidity for particular classes of games, or proving that rigid games exist that self-test particular quantum states.  Two-player games
that are known to be rigid include the CHSH game \cite{popescu1992states, mckague2012robust},  the magic square game \cite{Wu16}, the chained Bell inequalities \cite{Supic2016}, the Mayers-Yao criterion \cite{mayers1998quantum, Magniez:2006}, Hardy's test \cite{Rabelo2012},  
the Hadamard-graph coloring game \cite{mancinska2015maximally},
and various classes of binary games \cite{miller2013optimal, Wang2016, Bamps2015}.  New results on rigid games add to the  tools available
for protocols based on untrusted devices.

In the current paper we prove that the magic pentagram game (see Figure~\ref{fig:game}) is rigid.  This game is a natural one to
study: in particular, it was originally proposed alongside
the magic square game \cite{Mermin90}, and it shares some of the same properties that make the magic square
game useful in cryptography (in particular, it shares the property that an optimal strategy must yield a perfect shared key bit pair
between two parties, which was exploited in \cite{Jain:2017}).  From a resource standpoint, it also offers an improvement over the magic
square game: whereas the magic square game requires $9$ questions to self-test $2$ EPR pairs, we will prove that the magic pentagram
game self-tests $3$ EPR pairs with $20$ questions.  If we compare the number of bits of randomness needed to generate the questions set
to the number EPR pairs tested, the magic square has a ratio of $\frac1{2}\log_2{9} \approx 1.58$, while the magic pentagram
game has a ratio of $\frac{1}{3} \log_2{20} \approx 1.44$.

The optimal strategy for the magic pentagram game is shown in Figure~\ref{fig:strategy}.  
Our main result is summarized below, and proved formally in Propositions~\ref{prop:a_approx}, \ref{prop:b_approx},
\ref{prop:state_approx}, and Corollary~\ref{cor:approx}.

\begin{theorem}[Informal] Suppose that Alice and Bob have a strategy for the magic pentagram game that wins
\label{thm:informal}
with probability $1-\epsilon$.  Then, after the application of a local isometry on Alice's and Bob's systems, the following
statements hold.
\begin{enumerate}
\item The shared state is within Euclidean distance $O ( \sqrt{\epsilon} )$ from a state of the form $(\Phi^+)^{\otimes 3} \otimes \left| junk \right>$,
where $\Phi^+$ denotes a Bell state and $\left| junk \right>$ denotes an arbitrary bipartite state.  (Proposition~\ref{prop:state_approx}.)

\item The post-measurement states under Alice's and Bob's measurements are approximated (up to $O ( \sqrt{\epsilon} )$)
by the corresponding post-measurement states from the strategy in Figure~\ref{fig:strategy}.  (Propositions~\ref{prop:a_approx}--\ref{prop:b_approx} and Corollary~\ref{cor:approx}.)
\end{enumerate}
\end{theorem}

Our proof is self-contained and borrows techniques from previous papers on rigidity \cite{mckague2012robust,
McKague16, Wu16}.
One of the challenges for the magic pentagram game is that the first player may associate two different measurements to a single observable  --- for example, in Figure~\ref{fig:game}, Alice may use a different measurement for vertex $1$ depending on whether the context is $G$ or $D$.  (This does not occur
in the magic square game.)  Our early technical work addresses this fact  --- see Propositions~\ref{prop:change}--\ref{prop:comm}
and the discussion that follows.

The coefficients of the error terms $O ( \sqrt{\epsilon} )$ for Theorem~\ref{thm:informal} are not given explicitly, and optimizing
these coefficients is left as an open problem.  (Tracing through the steps of the current proof might yield coefficients in the thousands.)

In the larger picture, the magic square game and the magic pentagram game are examples of binary constraint satisfaction
games \cite{Cleve2014}.
Arkhipov~\cite{Arkhipov12} proved that a certain natural subclass of 
binary constraint satisfaction problems --- specifically, those that are based on XOR clauses where every variable
is in exactly two clauses --- are all in a precise sense reducible to the magic square game and the magic pentagram game.
This suggests that our result is a step towards a full classification of winning quantum strategies within this class.

\section{The magic pentagram game}\label{sec:game}
The  pentagram game is a binary constraint satisfaction game between two parties, Alice and Bob. Its rules can be defined, as its name suggests, on a pentagram hypergraph, see Fig.~\ref{fig:game}.  The five hyperedges of the pentagram (the clauses or contexts) are labeled 
$C, D, E, F, G$, and each contains four vertices.  The hyperedges are each assigned a value: $\ell ( C ) = \ell (D ) = \ell ( E ) = \ell ( F ) = 1$, and $\ell ( G ) = -1$.  
The rules of the games are as follows: 
\begin{itemize}
\item{A context $j$ is chosen and a vertex $v \in j$ is chosen (both uniformly at random).  The context $j$ is given
to Alice and the vertex $v$ is given to Bob.}
\item{Alice assigns either $+1$ or $-1$ to each vertex in the context $j$, and Bob assigns $+1$ or $-1$ to $v$.} 
\item{Alice and Bob can communicate and agree on a strategy prior to the beginning of the game, but are not allowed to communicate once the game has begun.} 
\end{itemize}
The game is won if the following two conditions both hold:
\begin{itemize}
\item{The product of the values returned by Alice is equal to the pre-assigned value $\ell ( j )$.
}
\item{Alice and Bob return the same value for $v$.}
\end{itemize}

\begin{figure}[t]
\centering
\includegraphics[width=\linewidth]{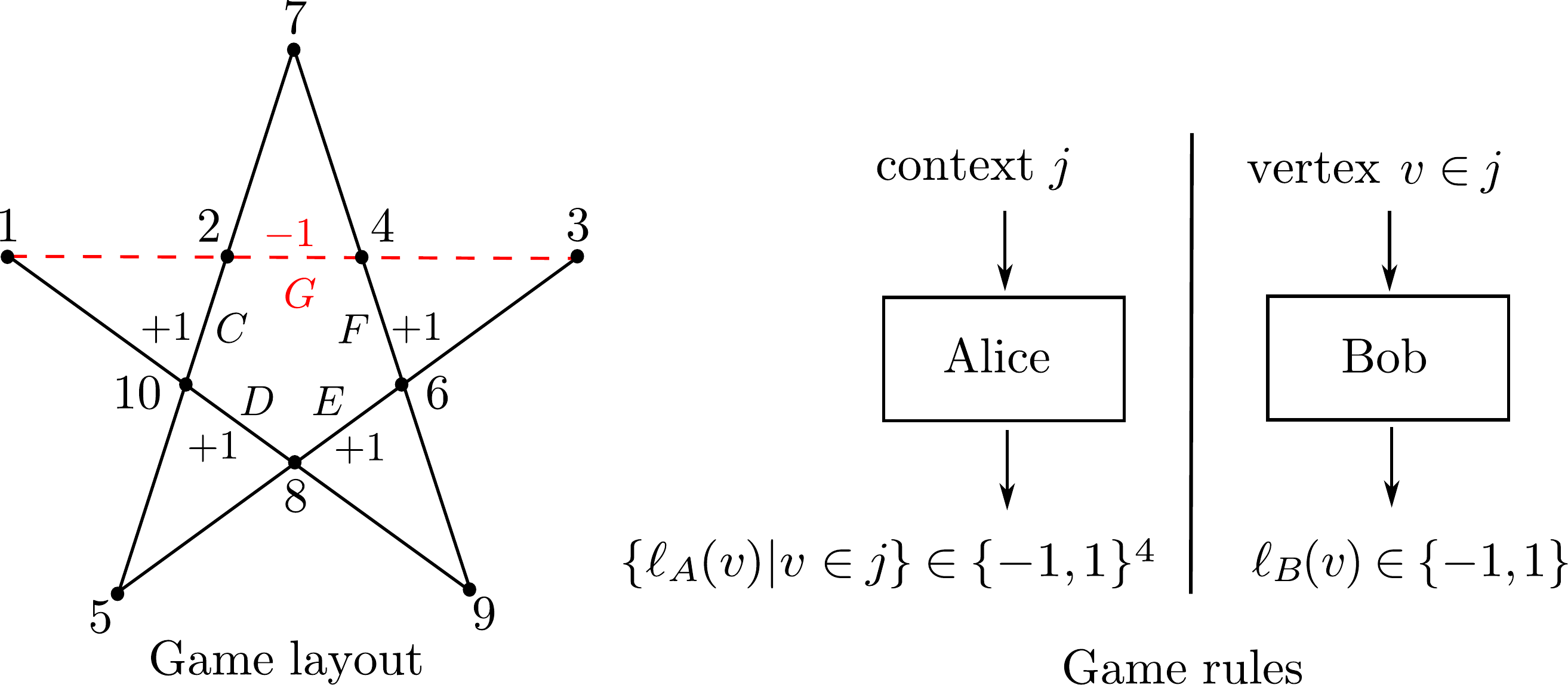}
\caption{\bf The pentagram game.}
\label{fig:game}
\end{figure}

There is no classical strategy to win this game perfectly, as is easily verified. However, it can be won with probability 1 using quantum resources~\cite{Mermin90,Mermin93}. A winning strategy is schematically shown in Fig.~\ref{fig:strategy}, with $Z, X$ and $I$ denoting the Pauli operators $\sigma_z,\sigma_x$, and the identity operator, respectively. They share six qubits, three at Alice's lab ($Q_1Q_2Q_3$) and three at Bob's ($Q_4Q_5Q_6$),  prepared in the maximally entangled state
\begin{equation}\label{3singlets}
\ket{\Phi^+}^{\otimes3}=\bigotimes_{i=1}^3\frac1{\sqrt{2}}(\ket{0}_{Q_i}\ket{0}_{Q_{i+3}}+\ket{1}_{Q_{i}}\ket{1}_{Q_{i+3}}),
\end{equation}
where $\ket{0},\ket{1}$ are the eigenbasis of the Pauli $Z$ operator. (When no confusion arises we drop the tensor product symbol and the subscript labels for Alice and Bob's subsystems.)  Upon receiving a hyperedge label $j$, Alice measures the four Pauli observables associated with the four vertices of $j$  on her three qubits, and then assigns to each vertex the value she obtains for the corresponding observable. These observables are reflection operators (i.e., Hermitian operators having eigenvalues in $\{ -1, +1 \}$) such that observables of adjacent vertices (vertices that are connected by the same hyperedge) all commute and thus can be measured simultaneously. Bob measures the observable of his input vertex on his three qubit system and assigns a $\{ -1, +1 \}$ value to the vertex according to the outcome of his measurement. By construction of this strategy, the winning conditions for this game, as listed above, are fulfilled for every input value $j$ and $v$.

We note that in this strategy any two non-adjacent observables  anti-commute.  (This will become important in later proofs.)
\begin{figure}[t]
\centering
\includegraphics[width=\linewidth]{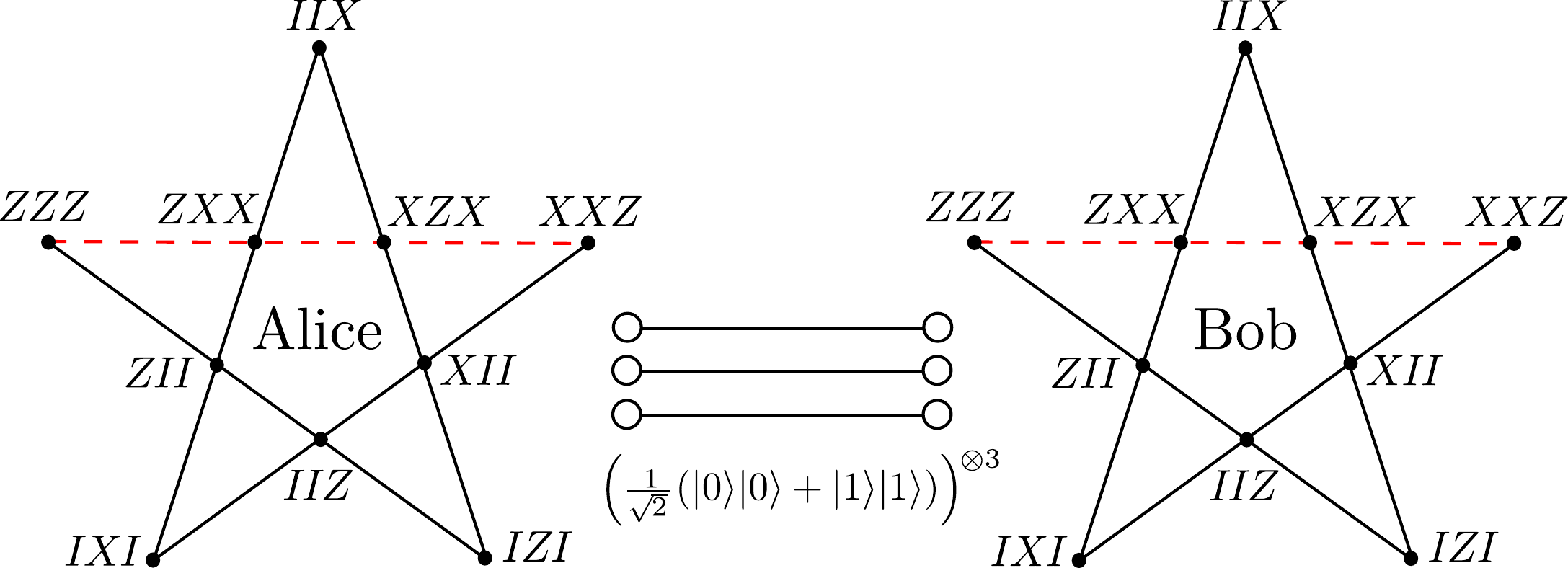}
\caption{{\bf A winning strategy.} }
\label{fig:strategy}
\end{figure}

\section{Strategies for the magic pentagram game}\label{sec:strategies}

Our goal is to relate arbitrary strategies for the magic pentagram game 
to the strategy in Figure~\ref{fig:strategy}.
The class of strategies that we study are captured in the following definition.

\begin{definition}
A \textit{projective strategy} for the magic pentagram game consists of the following data:
\begin{enumerate}
\item \textbf{The shared state:} Two finite-dimensional Hilbert spaces $\mathcal{H}_A$ and $\mathcal{H}_B$, and
a unit vector $\left| \psi \right> \in \mathcal{H}_A \otimes \mathcal{H}_B$.  

\item \textbf{Alice's measurements:} For each
$j \in \{ C, D, E, F, G \}$, a projective measurement $\{ M^j_t  \}$ on $\mathcal{H}_A$, where
$t$ varies over the set of all functions from $j$ to $\{ 0, 1 \}$ whose parity is equal to $\ell ( j)$.

\item \textbf{Bob's measurements:} For vertex $v$, a projective measurement
$\{ N_v^s \}_{s \in 0, 1 }$.
\end{enumerate}
\end{definition}
The functions obtained from these measurements specify the output values for Alice and Bob.  Note that we could have allowed for the shared state to be mixed and for the
measurements to be general positive-operator valued measures (POVMs).  However standard techniques
imply that any such strategy is a partial trace of one in the above form, so there is no generality lost. 

Additionally, we make the following definition.
\begin{definition}
\label{def:refl}
A \textit{reflection} is a Hermitian automorphism whose eigenvalues are contained in $\{ -1, +1 \}$.
A \textit{reflection strategy} for the magic pentagram game consists of the following data:
\begin{enumerate}
\item \textbf{The shared state:} Two finite-dimensional Hilbert spaces $\mathcal{H}_A$ and $\mathcal{H}_B$, and
a linear map $L \colon \mathcal{H}_B \to \mathcal{H}_A$ satisfying $\left\| L \right\|_2 = 1$.  

\item \textbf{Alice's reflections:} Reflections 
\begin{eqnarray*}
\{ R^j_{v} \mid j \in \{ C, D, E, F, G \}, v \in j \}
\end{eqnarray*}
on $\mathcal{H}_A$ 
such that the reflections that belong to any context $j$ all commute ($[R^j_{v},R^j_{v'}]=0$) and their product is equal to $\ell ( j ) I$.

\item \textbf{Bob's reflections:} Reflections $\{ S_v \}_v$ on $\mathcal{H}_B$.
\end{enumerate}
\end{definition}
Note that any projective strategy can be converted into a reflection strategy, and vice versa, via the relations
\begin{eqnarray}
R^j_v & = & \sum_{t( v ) = 0} M^j_t - \sum_{t( v ) = 1} M^j_t, \\
S_v & = & N_v^0 - N_v^1 \\
L & = & \langle\Phi^B\ket{\psi}
\end{eqnarray}
where $\ket{\Phi^B}=\sum_i \left| i i \right>$ on $\mathcal{H}_B$.

The probability distribution obtained from a projective measurement
$\{ O_1, \ldots, O_n \}$ on $\mathcal{H}_A$ is given by $( \left\| O_1 L \right\|_2^2 , \ldots , \left\| O_n L \right\|_2^2 )$,
and the probability distribution obtained from a projective measurement $\{ P_1 , \ldots, P_n \}$ 
on $\mathcal{H}_B$ is given by $( \left\| L  P_1 \right\|_2^2 , \ldots , \left\| L P_n \right\|_2^2 )$.
For any context $j$ and any vertex $v \in j$, the probability that Alice and Bob will assign different values to the vertex $v$ in a given reflection strategy is given by
\begin{eqnarray*}
\left\| \left( \frac{ I + R_v^j }{2} \right) L \left( \frac{ I - S_v }{2} \right) \right\|_2^2 + 
\left\| \left( \frac{ I - R_v^j }{2} \right) L \left( \frac{ I + S_v }{2} \right) \right\|_2^2.
\end{eqnarray*}
Thus the losing probability (that is, one minus the expected score) for the reflection strategy is
given by 
\begin{eqnarray*}
p_{\rm lose}& = & \frac{1}{20} \sum_{v \in j} \left(  
\left\| \left( \frac{ I + R_v^j }{2} \right) L \left( \frac{ I - S_v }{2} \right) \right\|_2^2 \right. \\
&& \left. +  \left\| \left( \frac{ I - R_v^j }{2} \right) L \left( \frac{ I + S_v }{2} \right) \right\|_2^2 \right) \\
& = & \frac{1}{20} \sum_{v \in j} \left\| L - R_v^j L S_v \right\|_2^2 \\
& = & \frac{1}{20} \sum_{v \in j} \left\| R_v^j L - L S_v \right\|_2^2.
\end{eqnarray*}

Thus we have the following.
\begin{proposition}
\label{prop:switch}
Let $( L, \left\{ R_v^j \right\}, \left\{ S_v \right\} )$ be a reflection strategy for 
the magic pentagram game which achieves winning probability $1 - \epsilon$.  Then,
for any context $j$ and vertex $v \in j$,
\begin{eqnarray}
\left\| R_v^j L - L S_v \right\|_2 & \leq & O ( \sqrt{ \epsilon } ),
\end{eqnarray}
\end{proposition}

Next we prove a series of properties for near-optimal strategies,
all of which are consequences of Proposition~\ref{prop:switch}.

\begin{proposition}[Changing contexts]
\label{prop:change}
Let
\begin{eqnarray*}
(L, \{ R_v^j \} , \{ S_v \} )
\end{eqnarray*}
be a reflection
strategy with expected score $1 - \epsilon$.  
Let $v_1, \ldots, v_n$ be a sequence of vertices
and $j_1, \ldots, j_n$ and $j'_1, \ldots, j'_n$ be sequences
of contexts such that $v_i \in j_i \cap j'_i$ for all $i$.  Then,
\begin{eqnarray*}
\left\| R_{v_1}^{j_1} R_{v_2}^{j_2} \cdots R_{v_n}^{j_n} L -
R_{v_1}^{j'_1} R_{v_2}^{j'_2} \cdots R_{v_n}^{j'_n} L \right\|_2 & \leq & O ( n \sqrt{\epsilon } ).
\end{eqnarray*}
\end{proposition}

\begin{proof}
Applying Proposition~\ref{prop:switch} inductively, we find that $R_{v_1}^{j_1}  \cdots R_{v_n}^{j_n} L$
and $R_{v_1}^{j'_1}  \cdots R_{v_n}^{j'_n} L$ are both within Euclidean distance
$O ( n \sqrt{\epsilon } )$ from $L S_{v_1} \cdots S_{v_n}$.
\end{proof}

The next two propositions certify the relation between reflection operators in a strategy with expected score $1 - \epsilon$. For convenience, hereafter we refer to sequences
$T_1, \ldots, T_n$ of matrices satisfying $\left\| T_{i+1} - T_i \right\|_2 \leq \delta$
as \textit{$\delta$-approximate sequences}. 
\begin{proposition}[Approximate commutativity]
\label{prop:comm}
Let $(L, \{ R_v^j \} , \{ S_v\} )$ be a reflection
strategy with expected score $1 - \epsilon$. Let $v$ and $w$ be adjacent vertices, such that $v,w\in j$, and let ${j'} \neq j$ be the other hyperedge which contains $w$.  Then,
\begin{eqnarray}
\left\| R_v^j R_w^{j'} L - R_w^{j'} R_v^j L \right\|_2 & \leq & O ( \sqrt{\epsilon } ) \\
\left\| L S_w S_v - L S_v S_w \right\|_2 & \leq & O ( \sqrt{\epsilon } ).
\end{eqnarray}
\end{proposition}

\begin{proof}
The desired result follows
easily by applications of Proposition~\ref{prop:switch}.
\end{proof}

Each vertex $v$ has two reflection operators for Alice
($R_v^j$ and $R_v^k$, where $j \cap k = \{ v \}$).  It is helpful
for some of the proofs that follow to single out one distinguished reflection operator
for each vertex.  We therefore make the following (arbitrary) assignments,
\begin{eqnarray}\label{R1toR10}
\begin{array}{ccc} 
R_1  :=  R_1^G & \hskip0.5in & R_6  :=  R_6^E \\
R_2  :=  R_2^G && R_7  :=  R_7^F \\
R_3  :=  R_3^E && R_8  :=  R_8^D \\
R_4  :=  R_4^F && R_9 := R_9^D \\
R_5  :=  R_5^E && R_{10} := R_{10}^C \,. 
\end{array}
\end{eqnarray}

\begin{proposition}[Approximate anti-commutativity]
\label{prop:anti}
Let $(L, \{ R_v^j \} , \{ S_v \} )$ be a reflection
strategy with expected score $1 - \epsilon$, 
and let $v\in j$ and $w\in j'$  be non-adjacent vertices (i.e., vertices
that never occur in the same context).  Then,
\begin{eqnarray}
\left\| R_v^j R_w^{j'} L + R_w^{j'} R_v^j L \right\|_2 & \leq & O ( \sqrt{\epsilon } ) \\
\label{santi}
\left\| L S_w S_v + L S_v S_w \right\|_2 & \leq & O ( \sqrt{\epsilon } ).
\end{eqnarray}
\end{proposition}
\begin{proof}
By Proposition~\ref{prop:change}, it suffices to prove these relations
with $R_v^j, R_w^{j'}$ replaced by $R_v, R_w$.  
We give a proof for $v = 7, w = 3$, which generalizes to cover
all other cases by symmetry.  The proof is inspired
by the proof of rigidity for the magic square game \cite{Wu16}.  Applying the rules for Alice's measurements
from Definition~\ref{def:refl} and the foregoing propositions, we find that the following
sequence is an $O( \sqrt{\epsilon} )$-approximate sequence:
\begin{eqnarray*}
&& R_7 R_3 L,  \\
&&(R_4 R_9 R_6 ) (R_6 R_8 R_5 ) L, \\
&&R_4 R_9 R_8 R_5 L, \\
&&R_4 (R_1 R_{10} R_8) R_8 R_5 L, \\
&&R_4 R_1 R_{10} R_5 L , \\
&&R_4 R_1 (R_2 R_7 ) L , \\
&&- R_3 R_7 L,
\end{eqnarray*}
and relation (\ref{santi}) follows similarly.
\end{proof}

The next proposition follows from Propositions~\ref{prop:switch}, \ref{prop:comm}, and \ref{prop:anti}.
\begin{proposition}
Let $v_1\in j_1, v_2\in j_2, \ldots, v_n\in j_n$ be a sequence of vertices
and $i \in \{ 1, 2, \ldots, n-1 \}$.  Then,
\begin{eqnarray*}
\left\| R_{v_1}^{j_1} \cdots R_{v_i }^{j_i} R_{v_{i+1} }^{j_{i+1}} \cdots R_{v_n}^{j_n} L \right. \\
\left. -b R_{v_1}^{j_1} \cdots  R_{v_{i+1} }^{j_{i+1}} R_{v_i }^{j_i} \cdots R_{v_n}^{j_n} L 
\right\|_2 & \leq & O ( n \sqrt{ \epsilon } ) \\
\left\| L S_{v_1} \cdots S_{v_i } S_{v_{i+1} } \cdots S_{v_n} \right. \\
\left. - b L S_{v_1} \cdots S_{v_{i+1} } S_{v_i } \cdots S_{v_n}  
\right\|_2 & \leq & O ( n \sqrt{ \epsilon } ),
\end{eqnarray*}
where $b = 1$ if $v_i, v_{i+1}$ are adjacent and $b = -1$ if $v_i, v_{i+1}$
are non-adjacent.  $\qed$
\end{proposition}

\section{Rigidity}\label{sec:rigidity}

In this section, we will use the following notation:
$Q_1, \ldots, Q_6$ will denote qubit registers
(each with a fixed isomorphism to $\mathbb{C}^2$).
The linear maps $H_i \colon Q_i \to Q_i$ denote the Hadamard
maps $\left| 0 \right> \mapsto \left| + \right>, \left| 1 \right> \mapsto \left| - \right>$,
and the linear maps $X_i, Z_i \colon Q_i \to Q_i$ denote the Pauli operators.
For any reflection $U$ on $\mathcal{H}_A \otimes \mathcal{H}_B$,
and $i \in \{ 1, 2, 3, 4, 5 , 6 \}$,
let the map
\begin{eqnarray}
C_i ( U ) \colon Q_i \otimes \mathcal{H}_A \otimes \mathcal{H}_B \to Q_i \otimes \mathcal{H}_A \otimes \mathcal{H}_B
\end{eqnarray}
denote the controlled operation $\left| 0 \right> \left< 0 \right| \otimes I + \left| 1 \right> \left< 1 \right| \otimes U$.
Note that these maps interact as follows:
\begin{eqnarray}
& X_i C_i ( U ) X_i =  C_i ( U ) U= U C_i ( U ) \\
& Z_i C_i ( U )  =  C_i ( - U ) = C_i (U ) Z_i
\end{eqnarray}

The next theorem asserts that some of the reflections in
a near-optimal strategy for the magic pentagram game
can be simulated by Pauli operators.  Let
\begin{eqnarray}
\begin{array}{ccc}
X'_1  =  R_6 & \hskip0.5in & X'_4 = S_6 \\
X'_2 = R_5 &  & X'_5 = S_5 \\
X'_3 = R_7 & & X'_6 = S_7 \\
Z'_1 = R_{10} & & Z'_4 = S_{10} \\
Z'_2 = R_9 && Z'_5 = S_9 \\
Z'_3 = R_8 && Z'_6 = S_8,
\end{array}
\end{eqnarray}
where the $R$s are given in Eq.~\eqref{R1toR10}.
These operators are chosen so that for $i\in\{1,2,3\}$ (and similarly for $i\in\{4,5,6\}$) the pairs $(X'_i, Z'_{i} )$, belong to non-adjacent vertices, while all the other pairs of operators  belong to adjacent vertices.  Thus
the approximate commutativity conditions and anti-commutativity
conditions are what one would expect for the corresponding Pauli operators. We note that the particular choice of the $X'$s and $Z'$s here is not unique. The following results will hold for any choice of $X'$s and $Z'$s as long as they satisfy the required approximate commutation relations.

\begin{proposition}
\label{prop:a_approx}
Let $(L, \{ R_v^j \} , \{ S_v \} )$ be a reflection
strategy with expected score $1 - \epsilon$.  Then, there
exists an isometry $\Psi_A$ from $\mathcal{H}_A$
to $\mathcal{H}_A \otimes Q_1 \otimes Q_2 \otimes Q_3$ such that for all $i \in \{ 1, 2 , 3 \}$,
\begin{eqnarray}
\left\| X_i \Psi_A L - \Psi_A X'_i L \right\|_2 & \leq & O ( \sqrt{\epsilon } ) \\
\left\| Z_i \Psi_A L - \Psi_A Z'_i L \right\|_2 & \leq & O ( \sqrt{\epsilon } ).
\end{eqnarray}
\end{proposition}

\begin{proof}
Our construction of the isometries follows previous papers on rigidity (e.g., \cite{McKague16}).
For each $i \in \{ 1, 2, 3 \}$ define
\begin{eqnarray}
\Psi_i \colon \mathcal{H}_A \to \mathcal{H}_A \otimes Q_i
\end{eqnarray}
by
\begin{eqnarray}
\label{psiexp}
\Psi_i ( z ) & = & [ C_i ( X'_i ) ] H_i [ C_i ( Z'_i )] (z \otimes \left| + \right> ).
\end{eqnarray}
Then, the following is an $O( \sqrt{\epsilon} )$-approximate sequence:
\begin{eqnarray*}
&& X_i \Psi_i L, \\
&& X_i [ C_i ( X'_i )] H_i [ C_i ( Z'_i ) ] (L \otimes \left| + \right>  ), \\
&& [C_i ( X'_i ) ] X'_i X_i H_i [ C_i ( Z'_i ) ] ( L \otimes \left| + \right> ), \\
&& [C_i ( X'_i ) ] H_i Z_i X'_i [ C_i ( Z'_i ) ] ( L \otimes \left| + \right> ), \\
&& [C_i ( X'_i ) ] H_i Z_i [ C_i ( - Z'_i ) ] X'_i ( L \otimes \left| + \right> ), \\
&& [C_i ( X'_i ) ] H_i [ C_i ( Z'_i ) ] X'_i ( L \otimes \left| + \right> ),  \\
&& \Psi_i X'_i L. 
\end{eqnarray*}
Thus, 
\begin{eqnarray*}
\left\| X_i \Psi_i L - \Psi_i X'_i L \right\|_2 & \leq & O ( \sqrt{\epsilon } ).
\end{eqnarray*}
Additionally, the following is an $O( \sqrt{\epsilon} )$-approximate sequence:
\begin{eqnarray*}
&& Z_i \Psi_i L, \\
&& Z_i [ C_i ( X'_i )] H_i [ C_i ( Z'_i ) ] (L \otimes \left| + \right>  ), \\
&& [ C_i ( X'_i )] Z_i H_i [ C_i ( Z'_i ) ] (L \otimes \left| + \right>  ), \\
&& [ C_i ( X'_i )] H_i X_i [ C_i ( Z'_i ) ] (L \otimes \left| + \right>  ), \\
&& [ C_i ( X'_i )] H_i  [ C_i ( Z'_i ) ] Z'_i X_i (L \otimes \left| + \right>  ), \\
&& [ C_i ( X'_i )] H_i  [ C_i ( Z'_i ) ] Z'_i  (L \otimes \left| + \right>  ), \\
&& \Psi_i Z'_i L.
\end{eqnarray*}
Thus,
\begin{eqnarray*}
\left\| Z_i \Psi_i L - \Psi_i Z'_i L \right\|_2 & \leq & O ( \sqrt{\epsilon } ).
\end{eqnarray*}

Also, if $i, k \in \{ 1, 2,3 \}$  with $k \neq i$, then
by Proposition~\ref{prop:comm}, the following is a $O( \sqrt{\epsilon} )$-approximate sequence:
\begin{eqnarray*}
&& X'_k \Psi_i L, \\
&& X'_k [ C_i ( X'_i )] H_i [ C_i ( Z'_i ) ] (L \otimes \left| + \right>  ), \\
&& X'_k [ C_i ( X'_i )] H_i L [ C_i ( Z'_{i+3} ) ] (I \otimes \left| + \right>  ), \\
&& X'_k [ C_i ( X'_i )] L H_i [ C_i ( Z'_{i+3} ) ] (I \otimes \left| + \right>  ), \\
&&  [ C_i ( X'_i )] X'_k L H_i [ C_i ( Z'_{i+3} ) ] (I \otimes \left| + \right>  ), \\
&&  [ C_i ( X'_i )] X'_k  H_i [ C_i ( Z'_i ) ] L  (I \otimes \left| + \right>  ), \\
&&  [ C_i ( X'_i )]   H_i [ C_i ( Z'_i ) ] X'_k L  (I \otimes \left| + \right>  ), \\
&&  \Psi_i X'_k L .
\end{eqnarray*}
Therefore
\begin{eqnarray}
\left\| X'_k \Psi_i L - \Psi_i X'_k L \right\|_2 & \leq & O ( \sqrt{\epsilon } )
\end{eqnarray}
and by similar reasoning,
\begin{eqnarray}
\left\| Z'_k \Psi_i L - \Psi_i Z'_k L \right\|_2 & \leq & O ( \sqrt{\epsilon } ).
\end{eqnarray}

Define $\Phi_i \colon \mathcal{H}_B  \to \mathcal{H}_B \otimes Q_i$ by the same expression
used to define $\Psi_i$, except with the operators $X'_i, Z'_i$ replaced with $X'_{i+3}, Z'_{i+3}$:
\begin{eqnarray}
\Phi_i ( z ) & = &   [ C_i ( X'_{i+3} ) ] H_i [ C_{i} ( Z'_{i+3} ) ] (z \otimes \left| + \right> ).
\end{eqnarray}
Then, $\left\| \Psi_i L - L \Phi_i \right\|_2 \leq O ( \sqrt{\epsilon} )$ by Proposition~\ref{prop:switch}.  
Let
\begin{eqnarray}
\Psi_A & = & \Psi_1 \Psi_2 \Psi_3.
\end{eqnarray}
Then, the following is an $O ( \sqrt{\epsilon} )$-approximate sequence:
\begin{eqnarray*}
&& X_2 \Psi_A L, \\
&& X_2 \Psi_1 \Psi_2 \Psi_3 L, \\
&& \Psi_1 X_2 \Psi_2 \Psi_3 L, \\
&& \Psi_1 X_2 \Psi_2 L \Phi_3,  \\
&& \Psi_1 \Psi_2 X'_2 L \Phi_3,  \\
&& \Psi_1 \Psi_2 X'_2  \Psi_3  L, \\
&& \Psi_1 \Psi_2 \Psi_3 X'_2 L.
\end{eqnarray*}
Therefore,
\begin{eqnarray}
\left\| X_2 \Psi_A L - \Psi_A X'_2 L \right\|_2 & \leq & O (\sqrt{\epsilon} ).
\end{eqnarray}
The desired result for $i=1, 3$ follows by similar reasoning.
\end{proof}

Likewise, we have the following.
\begin{proposition}
\label{prop:b_approx}
Let $(L, \{ R_v^j \} , \{ S_v \} )$ be a reflection
strategy with expected score $1 - \epsilon$.  Then, there
exists an isometry $\Psi_B$ from $\mathcal{H}_B$
to $\mathcal{H}_B \otimes Q_4\otimes Q_5\otimes Q_6$
such that for all $i \in \{ 4, 5, 6 \}$,
\begin{eqnarray}
\left\| L \Psi_B^\dagger X_i - L X'_i \Psi^\dagger_B \right\|_2 & \leq & O ( \sqrt{\epsilon } ) \\
\left\| L \Psi_B^\dagger Z_i - L Z'_i \Psi^\dagger_B \right\|_2 & \leq & O ( \sqrt{\epsilon } ).
\end{eqnarray}
\end{proposition}

\begin{proof}
Define $\Psi_i$ for $i \in \{ 4,5,6 \}$ by the same expression (\ref{psiexp}) that was used
in the previous proof, and let $\Psi_B = \Psi_4 \Psi_5 \Psi_6$.  The desired
result follows by the same reasoning that was used to prove  Proposition~\ref{prop:a_approx}.
\end{proof}

Note that Propositions~\ref{prop:a_approx} and \ref{prop:b_approx} easily generalize to
sequences of measurements --- for example, the following is an $O(\sqrt{\epsilon})$-approximate sequence:
\begin{eqnarray}
X_1 X_2 \Psi_A L, \\
X_1 \Psi_A X'_2 L, \\
X_1 \Psi_A L X'_5, \\
\Psi_A X'_1 L X'_5, \\
\Psi_A X'_1 X'_2 L.
\end{eqnarray}
Applying this method inductively, we have the following corollary.
\begin{corollary}
\label{cor:approx}
The isometries from Proposition~\ref{prop:a_approx} and \ref{prop:b_approx} satisfy the following.
For any sequence $M'_1, \ldots, M'_n \in \{ X'_1, X'_2, X'_3, Z'_1, Z'_2, Z'_3 \}$ and 
corresponding sequence $M_1, \ldots, M_n \in \{ X_1, X_2, X_3, Z_1, Z_2, Z_3 \}$,
\begin{eqnarray*}
\left\| M_1 \cdots M_n \Psi_A L - \Psi_A M'_1 \cdots M'_n  L \right\|_2 & \leq & O ( n \sqrt{\epsilon } ).
\end{eqnarray*}
For any sequence $N'_1, \ldots, N'_n \in \{ X'_4, X'_5, X'_6, Z'_4, Z'_5, Z'_6 \}$ and corresponding
sequence $N_1, \ldots, N_n \in \{ X_4, X_5, X_6, Z_4, Z_5, Z_6 \}$,
\begin{eqnarray*}
\left\| L \Psi_B^\dagger N_n \cdots N_1 - L  N'_n \cdots N'_1 \Psi_B^\dagger \right\|_2 & \leq & O ( n \sqrt{\epsilon } ). \qed
\end{eqnarray*}
\end{corollary}

Finally, we prove the following proposition, which addresses the image of
the $L$ under the isometry $\Psi_A \otimes \Psi_B$.  For each $i \in \{ 1, 2, 3 \}$,
let
\begin{eqnarray}
\phi_i^+ \colon Q_i \to Q_{i+3}
\end{eqnarray}
be defined by
\begin{eqnarray}
\phi_i^+ & = & \left[ \begin{array}{cc} \frac{1}{\sqrt{2}} & 0 \\ 0 & \frac{1}{\sqrt{2}} \end{array} \right].
\end{eqnarray}
(This is a matrix expression for an EPR pair.)  Let
\begin{eqnarray}
\phi_i^- & = & \left[ \begin{array}{cc} \frac{1}{\sqrt{2}} & 0 \\ 0 & - \frac{1}{\sqrt{2}} \end{array} \right] \\
\psi_i^+ & = & \left[ \begin{array}{cc} 0  & \frac{1}{\sqrt{2}} \\ \frac{1}{\sqrt{2}} & 0  \end{array} \right] \\
\psi_i^- & = & \left[ \begin{array}{cc} 0  & \frac{1}{\sqrt{2}} \\ - \frac{1}{\sqrt{2}} & 0  \end{array} \right].
\end{eqnarray}

\begin{proposition}
\label{prop:state_approx}
Let $L, \Psi_A, \Psi_B$
be the operators from Propositions~\ref{prop:a_approx} and
\ref{prop:b_approx}.  Then, for some $L' \colon \mathcal{H}_B \to \mathcal{H}_A$,
\begin{eqnarray}
\left\| L' \otimes \phi_1^+ \otimes \phi_2^+ \otimes \phi_3^+ - \Psi_A L \Psi^\dagger_B \right\|_2
& \leq & O ( \sqrt{ \epsilon } ).
\end{eqnarray}
\end{proposition}

\begin{proof}
Let $P = \Psi_A L \Psi^\dagger_B$.  By the score assumption,
\begin{eqnarray}
\left\| X'_i L X'_{i+3} - L \right\|_2 & \leq & O ( \sqrt{\epsilon } ) \\
\left\| Z'_i L Z'_{i+3} - L \right\|_2 & \leq & O ( \sqrt{\epsilon } ),
\end{eqnarray}
for $i \in \{ 1, 2, 3 \}$, therefore by Propositions~\ref{prop:a_approx} and \ref{prop:b_approx},
\begin{eqnarray}
\label{oneapprox}
\left\| X_i P X_{i+3} - P \right\|_2 & \leq & O ( \sqrt{\epsilon } ) \\
\label{twoapprox}
\left\| Z_i P Z_{i+3} - P \right\|_2 & \leq & O ( \sqrt{\epsilon } ),
\end{eqnarray}
Note that $X_i \phi_i^+ X_i = Z_i \phi_i^+ Z_i = \phi_i^+$, while
the other Bell states fail significantly to satisfy the same equalities:
\begin{eqnarray}
X_i \phi_i^- X_i & = & - \phi_i^- \\
Z_i \psi_i^+ Z_i & = & - \psi_i^+ \\
Z_i \psi_i^- Z_i & = & - \psi_i^-.
\end{eqnarray}
Write
\begin{eqnarray}
P & = & \sum_{v_1, v_2, v_3} v_1 \otimes v_2 \otimes v_3 \otimes P_{v_1, v_2, v_3},
\end{eqnarray}
where $v_i$ varies over $\{ \phi_i^+, \phi_i^-, \psi_i^+, \psi_i^- \}$.   Conditions (\ref{oneapprox}) 
and (\ref{twoapprox}) imply that all components $P_{v_1, v_2, v_3}$ except
$P_{\phi_1^+, \phi_2^+, \phi_3^+}$ must have Euclidean norm less than $O ( \sqrt{\epsilon} )$.  The desired result
follows.
\end{proof}

\vskip0.2in 

\section{Summary and Conclusions}\label{sec:conclusion}
Quantum rigidity allows a classical user to certify manipulations of quantum systems, thus enabling
quantum cryptography in a scenario in which the user does not trust her quantum apparatus (device-independent
quantum cryptography).  In this paper we have expanded the
toolbox for the device-independent setting by showing that the magic pentagram game is rigid.
In particular, this means that it is possible to certify the existence of $3$ ebits using a game that consists of only $20$ questions.

In our style of proof we have reduced some of the arguments for rigidity to bare manipulations of sequences
of operators (see the proofs in section~\ref{sec:rigidity}).  This style in particular allows us to cleanly handle conditions such as
approximate commutativity and anti-commutativity.  Such an
approach could be useful for proving more general results.  

A natural next step would be to try to parallelize our result (following \cite{reichardt2013classical, Ostrev:2015, Ostrev:2016, coladangelo2016parallel, natarajan2016robust, McKague16, mckague2016selfCHSH, chao2016test, coudron2016parallel}) to show that parallel copies of the magic
pentagram game can be used to certify a maximally entangled state of arbitrary size.  Then, we could
try to choose a small subset of the questions from the parallelized game and prove that that subset is adequate to
achieve rigidity.

The magic pentagram game is an example of a binary constraint satisfaction XOR game in which every variable appears
in exactly two contexts.  This class of games was studied in~\cite{Arkhipov12}, and the author proved that any
game in the class that exhibits pseudo-telepathy must in a sense contain either the magic square game
or the magic pentagram game (as topological minors of its relational graph). An interesting further direction would be to explore further the consequences for our rigidity
result (and \cite{Wu16}) for the class from \cite{Arkhipov12}.

\acknowledgements
The authors would like to thank  Cedric Lin for bringing Ref.~\cite{Arkhipov12} to our attention, and Matthew McKague for helpful technical discussions about our proofs. AK is funded by the US Department of Defense.

\bibliographystyle{apsrev4-1}
\bibliography{pentagram}

\end{document}